\documentclass{article}
\usepackage{latexsym,amssymb,amsmath,algorithm,algorithmicx,algpseudocode,graphicx,epsfig,subfigure}
\usepackage{amsthm}
\usepackage{multirow}
\usepackage{authblk}
\usepackage{mathtools}
\usepackage{color}

\usepackage[letterpaper,margin=1in]{geometry}
\long\def\longdelete#1{}

\newtheorem{theorem}{Theorem}[]

\newtheorem{definition}{Definition}

\newtheorem{invariant}{Invariant}
\newtheorem{lemma}{Lemma}

\def\max{{\rm max}}
\def\min{{\rm min}}

\title{\textbf{An $O(1)$-Approximation Algorithm for Dynamic Weighted Vertex Cover with Soft Capacity}}
\longdelete{
    \author[1]{Hao-Ting Wei \thanks{Department of Industrial Engineering and Engineering Management,
    National Tsing Hua University, Hsinchu 30013,Taiwan. Email: s104034526@m104.nthu.edu.tw}}
    \author[2]{Wing-Kai Hon\thanks {Department of Computer Science,
    National Tsing Hua University, Hsinchu 30013, Taiwan. Email: wkhon@cs.nthu.edu.tw}}
    \author[3]{Paul Horn\thanks {University of Denver, Denver, USA. Email: paul.horn@du.edu}}
    \author[1]{Chung-Shou Liao\thanks {Department of Industrial Engineering and Engineering Management,
    National Tsing Hua University, Hsinchu 30013,Taiwan. Email: csliao@ie.nthu.edu.tw}}
    \author[4]{\protect\\ Kunihiko Sadakane\thanks{Department of Mathematical Informatics,
    The University of Tokyo, Tokyo, Japan. Email: sada@mist.i.u-tokyo.ac.jp}}
    }

\author{%
Hao-Ting Wei\thanks{Department of Industrial Engineering and Engineering Management,
    National Tsing Hua University, Hsinchu 30013,\\Taiwan. \texttt{Email:s104034526@m104.nthu.edu.tw}.}%
\quad Wing-Kai Hon\thanks{Department of Computer Science,
    National Tsing Hua University, Hsinchu 30013, Taiwan. \texttt{Email:wkhon@cs.nthu.edu.tw}.}%
\quad Paul Horn\thanks{University of Denver, Denver, USA. \texttt{Email:paul.horn@du.edu}.}%
\\ \quad Chung-Shou Liao\thanks{Department of Industrial Engineering and Engineering Management,
    National Tsing Hua University, Hsinchu 30013,\\Taiwan. \texttt{Email:csliao@ie.nthu.edu.tw}.}%
\quad Kunihiko Sadakane\thanks{Department of Mathematical Informatics,
    The University of Tokyo, Tokyo, Japan. \texttt{Email:sada@mist.i.u-tokyo.ac.jp}.}%
}

    \date{}

\begin{document}
\maketitle
\setcounter{page}{0}
\thispagestyle{empty}
\begin{abstract}
This study considers the (\emph{soft}) capacitated vertex cover problem in a dynamic setting.
This problem generalizes the dynamic model of the vertex cover problem, which has been intensively studied in recent years.
Given a dynamically changing vertex-weighted graph $G=(V,E)$,
which allows edge insertions and edge deletions,
the goal is to design a data structure that maintains an approximate minimum vertex cover while satisfying the capacity constraint of each vertex.
That is, when picking a copy of a vertex $v$ in the cover, the number of $v$'s incident edges
covered by the copy is up to a given capacity of $v$.
We extend Bhattacharya et al.'s work [SODA'15 and ICALP'15]
to obtain a deterministic primal-dual algorithm for maintaining
a constant-factor approximate
minimum capacitated vertex cover with $O(\log n / \epsilon)$ amortized update time,
where $n$ is the number of vertices in the graph.
\longdelete{
The approximation algorithm
yields the same ratio
in a generalized dynamic
}
The algorithm can be extended to (1) a more general model in which
each edge is associated with a non-uniform and unsplittable demand,
and (2) the more general capacitated set cover problem.
\longdelete{

In particular,
we make the data structure more flexible and consider a generalized
model in which every edge is associated with a
nonuniform
demand, each to be assigned to an
incident
vertex.
}
\end{abstract}

\newpage

\section{Introduction}

Dynamic algorithms have received fast-growing attention in the past decades,
especially for some classical combinatorial optimization problems
such as connectivity~\cite{AT,DI,HL}, vertex cover, and maximum matching~\cite{BGS,BCH1,BCH2,BHI1,NS,OR,PS,SS}.
This paper focuses on the fully dynamic model of the vertex cover problem,
which has been intensively studied in recent years.
Given a vertex-weighted graph $G=(V,E)$ which is constantly updated
due to a sequence of edge insertions and edge deletions,
the objective is to maintain a subset of vertices $S \subseteq V$ at any given time, such that
every edge is incident to at least one vertex in~$S$ and
the weighted sum of $S$
is minimized.
We consider a generalization of the problem,
where each vertex is associated with a given capacity.
When picking a copy of a vertex $v$ in~$S$, the number of its incident edges that
can be covered by such a copy is bounded by $v$'s given capacity.
The objective is to find a (\emph{soft}) capacitated weighted vertex cover~$S$
with minimum weight, i.e. $\sum_{v \in S} c_v x_v$ is minimized,
as well as an assignment of edges
such that the number of edges assigned to a vertex $v$ in $S$ is at most $k_v x_v$,
where $c_v$ is the cost of~$v$, $k_v$ is the capacity of $v$,
and $x_v$ is the number of selected copies of $v$ in $S$.
The static model of this generalization is the so-called \emph{capacitated vertex cover} problem, introduced by Guha et al.~\cite{GHKO}.

\medskip

\noindent
{\bf Prior work.}
For the vertex cover problem in a dynamic setting,
Ivkovic and Lloyd~\cite{ZL}
presented the pioneering work
wherein their fully dynamic algorithm
maintains a 2-approximation factor to vertex cover
with $O((n + m)^{0.7072})$ update time,
where $n$ is the number of vertices and $m$ is the number of edges.
Onak and Rubinfeld~\cite{OR} designed a randomized data structure that
maintains a large constant approximation ratio with $O(\log^2 n)$ amortized update time in expectation;
this is the first result that achieves a constant approximation factor with
polylogarithmic update time.
Baswana, Gupta, and Sen~\cite{BGS} designed another randomized data structure which
improves the approximation ratio to two, and simultaneously improved
the amortized update time to $O(\log n)$.
Recently,
Solomon~\cite{SS} gave the currently best randomized algorithm, which
maintains a 2-approximate vertex cover with $O(1)$ amortized update time.

\begin{table*}[h]
\normalsize
\begin{center}
\renewcommand{\arraystretch}{1}
\begin{tabular}{|c|c|c|c|l|}
\hline
\multirow{ 2}{*}{\bf Problem}& {\bf Approx.}& \multirow{ 2}{*}{\bf Update Time}&{\bf Data} &\multicolumn{1}{c|}{\multirow{2}{*}{\bf Reference}}\\
&{\bf Guarantee} & &{\bf Structure}&\\
\hline
UMVC & $O(1)$ &$O(\log^2 n)$ amortized & randomized &\hspace{1pt} STOC'10 \cite{OR}\\
\hline
UMVC & 2 &$O(\log n)$ amortized & randomized&\hspace{1pt} FOCS'11 \cite{BGS}\\
\hline
UMVC & 2 &$O(1)$ amortized & randomized&\hspace{1pt} FOCS'16 \cite{SS}\\
\hline
\hline
UMVC & 2 &$O(\sqrt{m})$ worst-case& deterministic &\hspace{1pt} STOC'13 \cite{NS}\\
\hline
UMVC & $2+\epsilon$ & $O(\log n / \epsilon^2)$ amortized & deterministic &\hspace{1pt} SODA'15 \cite{BHI1}\\
\hline
UMVC  & $2+\epsilon$ & $O(\gamma / \epsilon^2)$ worst-case & deterministic &\hspace{1pt} SODA'16 \cite{PS}\\
\hline
UMVC & $2+\epsilon$ &$O(\log^3 n)$ worst-case& deterministic&\hspace{1pt} SODA'17 \cite{BCH1}\\
\hline
WMVC & $2+\epsilon$ & $O(\log n / \epsilon^2 )$ amortized & deterministic &\multicolumn{1}{c|}{This paper}\\
\hline
\hline
UMSC & $O(f^3)$ & $O(f^2)$ amortized & deterministic &\hspace{1pt} IPCO'17 \cite{BCH2}\\
\hline
WMSC & $O(f^2)$ & $O(f\log(n+m) / \epsilon^2)$ amortized & deterministic &\hspace{1pt} ICALP'15 \cite{BHI2}\\
\hline
\multirow{ 2}{*}{WMSC}& $O(f^3)$ & $O(f^2)$ amortized & \multirow{ 2}{*}{deterministic} &\hspace{1pt} \multirow{ 2}{*}{STOC'17 \cite{GKKP}}\\
& $O(\log n)$& $O(f\log n)$ amortized& &\\
\hline
\hline
WMCVC& $ O(1) $ & $O(\log n / \epsilon )$ amortized &deterministic &\multicolumn{1}{c|}{This paper}\\
\hline
WMCSC& $ O(f^2) $ & $O(f\log (n+m) / \epsilon )$ amortized &deterministic &\multicolumn{1}{c|}{This paper}\\
\hline
\end{tabular}
\end{center}
\caption{
Summary of results for
unweighted (resp. weighted) minimum vertex cover (UMVC (resp. WMVC)),
unweighted (resp. weighted) minimum set cover (UMSC (resp. WMSC)), where $f$ is the maximum frequency of an element,
and weighted minimum capacitated vertex (resp. set) cover
(WMCVC (resp. WMCSC))
}\label{table1}
\end{table*}

For deterministic data structures,
Onak and Rubinfeld~\cite{OR} presented a data structure that maintains an
$O(\log n)$-approximation algorithm with $O(\log^2 n)$ amortized update time.
Bhattacharya~et~al.~\cite{BHI1}
proposed the first deterministic data structure that maintains a
constant ratio, precisely, a $(2+\epsilon)$-approximation to vertex cover
with polylogarithmic $O(\log n / \epsilon^2)$ amortized updated time.
Existing work also considered the worst-case update time.
Neiman and Solomon~\cite{NS} provided a 2-approximation dynamic algorithm
with $O(\sqrt{m})$ worst-case update time.
Later, Peleg and Solomon~\cite{PS}
improved the worst-case update time to $O( \gamma / \epsilon^2)$,
where $\gamma$ is the arboricity of the input graph.
Very recently,
Bhattacharya et al.~\cite{BCH1} extended their hierarchical data structure to
achieve the currently best worst-case update time of $O(\log^3 n)$.
Note that the above studies only discussed the unweighted vertex cover problem,
the objective of which is to find a vertex cover with minimum cardinality.

Consider the dynamic (weighted) set cover problem.
Bhattacharya et al.~\cite{BHI2} used a hierarchical data structure similar to that reported in~\cite{BHI1},
and achieved a scheme with $O(f^2)$-approximation ratio and
$O(f \log(n+m) / \epsilon^2 )$ amortized updated time,
where $f$ is the maximum frequency of an element.
Very recently, Gupta et al.~\cite{GKKP} improved the amortized update time to
$O(f^2)$,
albeit the dynamic algorithm achieves a higher approximation ratio of~$O(f^3)$.
They also offered another $O(\log n)$-approximation dynamic algorithm in
$O(f\log n)$ amortized update time.
Bhattacharya~et~al.~\cite{BCH2}
simultaneously derived the same outcome with $O(f^3)$-approximation ratio and $O(f^2)$ amortized update time
for the unweighted set cover problem.
Table~\ref{table1} presents a summary of the above results.

\medskip

\noindent
{\bf Our contribution.}
In this study we investigate the (\emph{soft}) capacitated vertex cover problem in the dynamic setting,
where there is no bound on the number of copies of each vertex that can be selected.
We refer to
the primal-dual technique reported in~\cite{GHKO},
and present the first deterministic algorithm for this problem, which
can maintain
an $O(1)$-approximate minimum capacitated (weighted) vertex cover
with $O(\log n / \epsilon)$ amortized update time.
The algorithm can be extended to a more general model in which
each edge is associated with a given demand, and the demand has to be assigned to an
incident vertex.
That is, the demand of each edge is \emph{nonuniform}
and \emph{unsplittable}.
Also, it can be extended to solve the more general capacitated set cover problem,
where the input graph is a hyper-graph, and each edge may connect to multiple vertices.

The proposed dynamic mechanism builds on Bhattacharya et al.'s $(\alpha,\beta)$-partition structure~\cite{BHI1,BHI2},
but a \emph{careful adaptation} has to be made to cope with the newly introduced capacity constraint.
Briefly, applying the \emph{fractional matching} technique in Bhattacharya~et~al.'s algorithm
cannot directly lead to a constant approximation ratio
in the capacitated vertex cover problem.
The crux of our result is the re-design of a key parameter, \emph{weight} of a vertex, in the dual model.
Details of this modification are shown in the next section.

\longdelete{ 
It's noted that the proposed primal-dual dynamic algorithm can be simply extended to the capacitated \emph{set cover} problem
with $O(f^2)$-approximation ratio in $O(f\log (n+m) / \epsilon)$ amortized update time,
where $f$ is the maximum frequency.
}

In addition, if we go back to the original vertex cover problem without capacity constraint,
the proposed algorithm is able to resolve the \emph{weighted} vertex cover problem
by maintaining a $(2+\epsilon)$-approximate weighted vertex cover
with $O(\log n / \epsilon^2)$ amortized update time.
This result achieves the same approximation ratio as the algorithm in~\cite{BHI1},
but they considered the unweighted model.
Details of
this discussion are presented in the end of Section~3.

\subsection{Overview of our technique}

First,
we recall the mathematical model of the capacitated vertex cover problem which was first introduced by Guha et al.~\cite{GHKO}.
In this model, $y_{ev}$ serves as a binary variable that indicates whether an edge $e$ is covered by a vertex $v$.
Let $N_v$ be the set of incident edges of $v$, $k_v$ and $c_v$ be the capacity and the cost of a vertex $v$, respectively.
Let $x_v$ be the number of selected copies of a vertex $v$.
An 
integer program (IP) model of the problem can be formulated as follows (the minimization program on the left):

\longdelete{
\[\begin{array}{lllr}
{\bf Min} & \sum_{v} c_v x_v && \\
&&&\\
{\bf s.t}
& y_{ev} + y_{eu} \geq 1,  & \quad \forall e=\{u,v\}\in E & \\
& k_v x_v - \sum_{e\in N_v} y_{ev} \geq 0, & \quad \forall v\in V  & \\
& x_v \geq y_{ev}, & \quad \forall v\in e, \forall e \in E  & \\
& y_{ev}\in \{ 0, 1 \}, & \quad \forall v\in e, \forall e \in E  & \\
& x_{v}\in \mathbb{N}, & \quad \forall v\in V & \\
\end{array}
\]
}

\begin{minipage}[c]{0.49\linewidth}
\begin{small}
\[\begin{array}{lllr}\rule{-0.8cm}{0ex}
{\bf Min} & \sum_{v} c_v x_v && \\
&&&\\
\rule{-0.8cm}{0ex}
{\bf s.t}
& y_{ev} + y_{eu} \geq 1,  & \forall e=\{u,v\}\in E & \\
\rule{-0.8cm}{0ex}
& k_v x_v - \sum_{e\in N_v} y_{ev} \geq 0, & \forall v\in V  & \\
\rule{-0.8cm}{0ex}
& x_v \geq y_{ev}, &  \forall v\in e, \forall e \in E  & \\
\rule{-0.8cm}{0ex}
& y_{ev}\in \{ 0, 1 \}, &  \forall v\in e, \forall e \in E  & \\
\rule{-0.8cm}{0ex}
& x_{v}\in \mathbb{N}, &  \forall v\in V & \\
\end{array}
\]
\end{small}
\end{minipage}
\hspace*{-2ex}
\rule[-1.4cm]{.1pt}{3cm}
\begin{minipage}[c]{0.5\linewidth}
\begin{small}
\[\begin{array}{lllr}
{\bf Max} & \sum_{e \in E} \pi_e && \\
&&&\\
\rule{0.3cm}{0ex}
{\bf s.t}
\rule{0.3cm}{0ex}
& k_v q_v + \sum_{e\in N_v} l_{ev} \leq c_v, & \forall v \in V & \\
\rule{0.3cm}{0ex}
& q_v + l_{ev} \geq \pi_e, & \forall v\in e, \forall e \in E  & \\
\rule{0.3cm}{0ex}
& q_v \geq 0, & \forall v \in V & \\
\rule{0.3cm}{0ex}
& l_{ev} \geq 0, &  \forall v\in e,  \forall e \in E & \\
\rule{0.3cm}{0ex}
& \pi_e \geq 0, & \forall e\in E & \\
\end{array}
\]
\end{small}
\end{minipage}
If we allow a relaxation of the above primal form,
i.e., dropping the integrality constraints,
its dual problem yields
a maximization problem.
The linear program for the dual can be formulated as shown in the above (the maximization program on the right;  also see~\cite{GHKO}).
One may consider this as a variant of the \emph{packing} problem,
where we want to pack a value of $\pi_e$ for each edge $e$,
so that the sum of the packed values is maximized.
Packing of $e$ is limited by the sum of $q_v$ and $l_{ev}$,
where $q_v$ is the \emph{global} ability of a vertex $v$ emitted to $v$'s incident edges,
and $l_{ev}$ is the \emph{local} ability of $v$ distributed to its incident edge $e$.
In this study,
we incorporate
the above IP model with its LP relaxation
for capacitated vertex cover into
the dynamic mechanism
proposed
by Bhattacharya et al.~\cite{BHI1,BHI2}.
They devised the
\emph{weight} of a vertex $v$ (in the dual model), denoted by $W_v$,
to obtain a feasible solution in the dual problem.
They also allowed a flexible range for $W_v$ to quickly adjust the solution
for dynamic updates while preserving its approximation quality.
Due to the additional capacity constraint in our problem,
a new \emph{weight} function is obviously required.
\longdelete{
Furthermore,
we optimize the setting of Bhattacharya et al.'s $(\alpha,\beta)$-partition structure
and show that the approximation quality of our result cannot be improved
based on the same dynamic mechanism.
}

\medskip

\noindent
{\bf Technical challenges.}
There are two major differences between
our
algorithm and
Bhattacharya et al.'s~\cite{BHI1,BHI2}.
First,
the capacity constraint in the primal problem leads to
the two variables $q_v$ and $l_{ev}$ in the dual problem in which we have to balance their values
when approaching $c_v$ to maximize the dual objective.
By contrast, the previous work considered one dual variable $l_{ev}$ without the restriction on the \emph{coverage} of a vertex.
We thus re-design~$W_v$, the \emph{weight} of a vertex $v$ to specifically
consider the capacitated scenario.
Yet,
even with the new definition of $W_v$,
there is still a second challenge on how to approximate the solution within a constant factor in the dynamic environment.
In order to achieve $O(\log n)$ amortized update time,
Bhattacharya et al.'s \emph{fractional matching} approach assigns the value of all $v$'s incident edges to $v$,
which, however, may result in a non-constant $h$, hidden in the approximation ratio, where $h$
is the largest number of copies selected in the cover.
We observe that we cannot remove $h$ from the approximation guarantee based on the $(\alpha,\beta)$-partition structure
if we just select the minimum value of $\alpha$, as it is done in~\cite{BHI1,BHI2}.
The key insight is that we show a bound on the value of $\alpha$, which restricts the updates of the dynamic mechanism.
With the help of this insight,
we are able to revise the setting of~$\alpha$ to derive a constant approximation ratio,
while maintaining the $O(\log n)$ update time.

\longdelete{ 
Admittedly, we extend Bhattacharya et al.'s~\cite{BHI1,BHI2} idea of computing the \emph{weight} of a vertex $v$ in the dual problem.
They also permit a range of the weight of a vertex $v$, $W_v \in (c_v/\varepsilon, c_v]$ for the purpose of updating even more quickly,
where $\varepsilon$ is a constant.
}

\section{Level Scheme and its Key Property}
The core of Bhattacharya et al.'s $(\alpha,\beta)$-partition structure~\cite{BHI1,BHI2} is
a \emph{level scheme}~\cite{OR}
that is used to maintain a feasible solution in their dual problem.
In this section, we demonstrate (in a different way from the original papers)
how this scheme can be applied to our dual problem, and describe the key property that the scheme guarantees.

A \emph{level scheme} is an assignment $\ell: V\rightarrow\{ 0, 1, \ldots, L\}$ such that
every vertex $ v\in V$ has a level~$\ell(v)$.  Let $c_\min$ and $c_\max$ denote the minimum and maximum costs of a vertex, respectively.
For our case, we set $L = \lceil \log_{\beta}( n\mu\alpha / c_{\min})\rceil$
for some $\alpha, \beta > 1$ and $\mu > c_{\max}$.   Based on $\ell$, each edge $(u,v)$ is also associated with a level $\ell(u,v)$,
where $\ell(u,v) = \max \{\ell(u),\ell(v)\}$.  An edge is assigned to the higher-level endpoint, and ties are broken arbitrarily
if both endpoints have the same level.

Each edge $(u,v)$ has a weight $w(u,v)$ according to its level, such that $w(u,v) = \mu\beta^{-\ell(u,v)}$.
Each vertex $v$ also has a weight $W_v$, which is defined based on the incident edges of $v$ and their corresponding levels.
Before giving details on $W_v$, we first define some notations.
Let $N_{v}=\{u \mid (u,v)\in E \} $ be the set of vertices adjacent to $v$  (i.e., the neighbors of $v$).
Let $N_{v}(i)$ denote the set of level-$i$ neighbors of $v$,
and $N_{v}(i, j)$ denote the set of $v$'s neighbors whose levels are in the range $[i, j]$.
That is, $N_{v}(i)=\{u \mid (u,v)\in E \land \ell(u) = i\}$ and $N_{v}(i,j)=\{u \mid (u,v)\in E \land \ell(u) \in [i,j]\}$.
The \emph{degree} of a vertex $v$ is denoted by
$D_{v}=|N_{v}|$. Similarly, we define $D_{v}(i)=|N_{v}(i)|$ and $D_{v}(i,j)=|N_{v}(i, j)|$.
Finally, we use $\delta(v)$ to denote the set of edges assigned to a vertex $v$.
Now, the weight $W_v$ of a vertex $v$ is defined as follows:

\begin{description}
  \item[Case 1] $D_v(0,\ell(v))>k_v$:
  \begin{align*}
  W_v =  k_v \mu\beta^{-\ell(v)}+\sum_{i>\ell(v)}{\min\{k_v,D_v(i)\}\mu\beta^{-i}}
  \end{align*}
  \item[Case 2] $D_v(0,\ell(v))\leq k_v$:
  \begin{align*}
  W_v = D_v(0,\ell(v))\mu\beta^{-\ell(v)}+\sum_{i>\ell(v)}{\min\{k_v,D_v(i)\}\mu\beta^{-i}}
  \end{align*}
\end{description}

Due to the capacity constraint,
we consider whether the number of level-$i$ neighbors of~$v$, $0 \le i \le \ell(v)$,
is larger than the capacity of~$v$, to define the weight of a vertex~$v$.
Note that the total weight of the edges that are assigned to~$v$ or incident to~$v$
can contribute at most $k_v w(u,v)$ to $W_v$.
Briefly, the weight of a vertex has two components:  one that is dependent on the incident edges with level
$\ell(v)$,
and the other that is dependent on the remaining incident edges.
For convenience, we call the former component $\mathit{Internal}_v$ and the latter component as $\mathit{External}_v$.
Moreover, we have:
\[
\mathit{External}_v\quad \leq\quad k_v\sum_{i>\ell(v)}{\mu\beta^{-i}}\ \ \leq\ \ (1/(\beta-1))k_v \mu\beta^{-\ell(v)}.
\]

\longdelete{ 
The edges which assigned to the vertex $v$ can at most contribute to $k_v$ edge's weight.
On top of that, for those incident edges which edge level
$\ell(j) > \ell(v)$ it also dedicates to a portion of the weight of a vertex $v$.
However, for every level in which it can only contribute to at most $k_v$ edge's weight.
}

In general, an arbitrary level scheme cannot be used to solve our problem.  What we need is
a \emph{valid} level scheme, which is defined as follows.
\begin{definition}
A level scheme is \emph{valid} if $W_v\leq c_v$, for every vertex $v$.
\end{definition}
\begin{lemma} \label{valid}
Let $V_0$ denote the set of level-$0$ vertices in a valid level scheme.
Then, $ V\setminus V_0$ forms a \emph{vertex cover} of $G$.
\end{lemma}
\begin{proof}
Consider any edge $(u,v)\in E$. We claim that at least one of its endpoints must be in $V\setminus V_0$.
Suppose that the claim is false which implies that $\ell(u) = \ell(v) = 0$ and $w(u,v) = \mu > c_{\max}$.
Since $w(u,v)$ appears in $\mathit{Internal}_v$, we have $W_v \geq w(u,v)$.
As a result, $c_v \geq W_v \geq \mu > c_{\max}$, which leads to a contraditction.
The claim thus follows, and so does the lemma.
\end{proof}

The above lemma implies that no edge is assigned to any level-$0$ vertex.
In our mechanism, we will maintain a valid level scheme, based on which each vertex in $V \setminus V_0$ picks enough copies to cover all the edges assigned to it;  this forms a valid capacitated vertex cover.

Next, we define the notion of \emph{tightness}, which is used to measure how good a valid level scheme performs.
\begin{definition}
A valid level scheme with an associated edge assignment is \emph{$\varepsilon$-tight}
if for every vertex~$v$ with $|\delta(v)| > 0$,  $W_v \in (c_v/\varepsilon, c_v]$.
\end{definition}

\begin{lemma} \label{vareps-tight}
Given an $\varepsilon$-tight valid level scheme,  we can obtain an $\varepsilon(2(\beta/(\beta-1))+1)$-approximation solution to the \emph{weighted minimum capacitated vertex cover (WMCVC)} problem.
\end{lemma}
\begin{proof}
First, we fix an arbitrary edge assignment that is consistent with the given valid level scheme.
For each vertex $v$ with $|\delta(v)| > 0$, we pick $\lceil |\delta(v)|/k_v \rceil$ copies to cover all the $|\delta(v)|$
edges assigned to it.   To analyze the total cost of this capacitated vertex cover,
we relate it to the value $\sum_{e} \pi_e$ of a certain feasible solution of the dual problem, whose
corresponding values of $q_v$ and $l_{ev}$ are as follows:

\begin{description}
   \item For every vertex $v$:
   \item $\bullet$ if $\lceil|\delta(v)|/k_v\rceil > 1$:   $q_v = \mu \beta^{-\ell(v)}$,  and $l_{ev} = 0$;
   \item $\bullet$ if $\lceil|\delta(v)|/k_v\rceil \leq 1$:   $q_v = \mu \sum_{ i \mid D_v(i) > k_v }  \beta^{-i}$,
                                                                                  $l_{ev} = 0$ if $D_v(\ell(e)) > k_v$, and $l_{ev} = \mu \beta^{-\ell(e)}$ otherwise.
   \item For every edge $e$:  $\pi_e = \mu \beta^{-\ell(e)}$.
\end{description}

\noindent
It is easy to verify that the above choices of $q_v$, $l_{ev}$, and $\pi_e$ give a feasible solution to the dual problem.

For the total cost of our solution, we separate the analysis into two parts, based on the multiplicity of the vertex:

\begin{description}
  \item [Case 1] $\lceil|\delta(v)|/k_v\rceil> 1$:
                             In this case, the external component of $W_v$ is at most $1/(\beta-1)$
                             of the internal component, so $W_v \leq (\beta/(\beta-1))k_v q_v$.
                             Then, the cost of all copies of $v$ is:
  \longdelete
  {
  We let $q_v = \pi_e $ and $l_{ev} = 0$.
  Due to the formula of computing a vertex weight, we have $W_v = k_v\mu\beta^{-\ell(v)}$+ External part.
  We connect $W_v$ to the feasible constraint in the dual problem such that $k_v q_v +\sum_{e\in N_v}{lev}\leq c_v$ and we know that $W_v\leq(\beta/(\beta-1))k_v q_v$.

  thus, we have the following analysis:}
  \begin{align*}
\lceil|\delta(v)|/k_v\rceil\cdot c_v &\leq \lceil|\delta(v)|/k_v\rceil\cdot \varepsilon \cdot W_v\\
    &\leq 2 \cdot \frac{ |\delta(v)|}{k_v} \cdot \varepsilon \cdot (\beta/(\beta-1)) k_v q_v\
        =\ 2\varepsilon(\beta/(\beta-1))\cdot\sum_{e\in\delta(v)}{\pi_e}.
  \end{align*}
  \item[Case 2] $\lceil|\delta(v)|/k_v\rceil\ =  1$:   In this case, we pick one copy of vertex $v$, whose cost is:
  \begin{align*}
  c_v&\leq \varepsilon\cdot W_v\ \leq \varepsilon \cdot \sum_{e \sim v} \pi_e\ =\ \varepsilon \cdot \left(\sum_{e\in\delta(v)}{\pi_e}+\sum_{e\notin\delta(v),\, e \sim v}{\pi_e} \right),
\end{align*}
where $e \sim v$  denotes $e$ is an edge incident to $v$.
\end{description}

In summary, the total cost is bounded by
\begin{align*}
&\sum_{v}{\left(\max\{\varepsilon,2\varepsilon(\beta/(\beta-1))\}\sum_{e\in\delta(v)}{\pi_e}+\varepsilon \sum_{e\notin\delta(v),\, e\sim v}{\pi_e}\right)}\\
&=\sum_{v}{\left( 2\varepsilon(\beta/(\beta-1))\sum_{e\in\delta(v)}{\pi_e}+\varepsilon \sum_{e\notin\delta(v),\, e\sim v}{\pi_e} \right)}\\
&=\varepsilon(2(\beta/(\beta-1))+1)\sum_{e}{\pi_e}\\
&\leq \varepsilon(2(\beta/(\beta-1))+1)\cdot OPT,
\end{align*}
where $OPT$ denotes the optimal solution of the dual problem, which is also a lower bound of the cost of any weighted capacitated vertex cover.
\end{proof}

The next section discusses how to dynamically maintain an $\varepsilon$-tight level scheme, for some constant factor $\varepsilon$ and with amortized $O(\log n/\epsilon)$ update time.
Before that, we show a greedy approach to get a $(\beta+1)$-tight level scheme to the static problem as a warm up.

First, we have the following definition.
\begin{definition}\label{improve}
A valid level scheme $\lambda$ is \emph{improvable} if some vertex can drop its level to get another level scheme $\lambda'$ such that $\lambda'$ is valid;
otherwise, we say $\lambda$ is \emph{non-improvable}.
\end{definition}

\begin{lemma}\label{non-improvable}
If a valid level scheme $\lambda$ is non-improvable, then $\lambda$ is $(\beta+1)$-tight.
\end{lemma}
\begin{proof}
To prove this lemma, we compare the weight $W_v$ of a vertex $v$ when its level is set to $i$ and $i+1$ which denoted by $W_{v}(i)$ and $W_{v}(i+1)$, respectively
(while the level of every other vertex remains unchanged).
\begin{description}
\item[Case 1] $D_{v}(0,i+1) \leq k_{v}$:
\begin{align*}
W_{v}(i+1) &=  D_{v}(0,i+1)\mu\beta^{-(i+1)}+\sum_{j>i+1}{\min\{k_v,D_v(j)\}\mu\beta^{-j}};\\
W_{v}(i) &= D_{v}(0,i)\mu\beta^{-(i)}+D_{v}(i+1)\mu\beta^{-(i+1)}+\sum_{j>i+1}{\min\{k_v,D_v(j)\}\mu\beta^{-j}}\\
&\leq \beta W_{v}(i+1);
\end{align*}
\item[Case 2] $D_{v}(0,i+1) > k_{v}$:
\begin{align*}
W_{v}(i+1) &= k_{v}\mu\beta^{-(i+1)} + \sum_{j>i+1}{\min\{k_v,D_v(j)\}\mu\beta^{-j}}; \\
W_{v}(i) &= \min\{k_v, D_{v}(0,i)\}\mu\beta^{-(i)}+D_{v}(i+1)\mu\beta^{-(i+1)}+\sum_{j>i+1}{\min\{k_v,D_v(j)\}\mu\beta^{-j}}\\
&\leq (\beta+1) W_{v}(i+1);
\end{align*}
\end{description}
In both cases, the weight $W_v(i)$ is at most $(\beta+1)$ times of $W_{v}(i+1)$.   Thus, if a vertex cannot drop its level,
either its current level is $0$, or by doing so we have $W_v(\ell(v)-1) > c_v$;   the latter implies that the current value of
$W_v = W_v(\ell(v))$ is larger than $c_v/(\beta+1)$.
Thus, if no vertex can drop its level, then the level scheme is $(\beta+1)$-tight.
\end{proof}

If we set the level of every vertex to $L$ initially, it is easy to check that by our choice of $L$ as $\lceil \log_{\beta}( n\mu\alpha / c_{\min})\rceil$,
such a level scheme is valid.
Next, we examine each vertex one by one, and drop its level as much as possible while the scheme remains valid.
In the end, we will obtain a non-improvable scheme, so that by the above lemma, the scheme is $(\beta+1)$-tight.
This implies a $(\beta+1)(2(\beta/(\beta-1))+1)$-approximate solution for the WMCVC problem.

\section{Maintaining an $\alpha(\beta+1)$-tight Level Scheme Dynamically}
In this section, we present our $O(1)$-approximation algorithm for the WMCVC problem, with amortized $O(\log n)$ update time
for each edge insertion and edge deletion.  We first state an invariant that is maintained throughout by our algorithm, and show how the latter is done.
Next, we analyze the time required to maintain the invariant with the potential method, and show that our proposed method can be updated efficiently as desired.
To obtain an $O(\log n)$ amortized update time, we relax the flexible range of the weight of a vertex $W_v$ by multiply a constant $\alpha$.
Let $c_v^\ast$ be $c_v/\alpha(\beta+1)$.  The invariant that we maintain is as follows.

\begin{invariant}\label{inv}
(1) For every vertex $v \in V \setminus V_0$, it holds that $c_v^\ast \leq W_v \leq c_v$, and (2)
for every vertex $v \in V_0$, it holds that $W_v \leq c_v$ .
\end{invariant}

By maintaining the above invariant,
we will automatically obtain an $\alpha(\beta+1)$-tight valid scheme.
As mentioned, we will choose a value for $\alpha$ in order to remove $h$ from the approximation ratio.
In particular, we will set $\alpha = (2\beta+1)/\beta+2\epsilon$, where $0 < \epsilon < 1$ to balance the update time, and $\beta = 2.43$ to minimize the approximation ratio,
so that we achieve the following theorem.

\begin{theorem}\label{Big}
There exists a dynamic level scheme $\lambda$ which can achieve a constant approximation ratio $(\approx36)$ for the WMCVC problem with $O(\log n/\epsilon)$ amortized update time.
\end{theorem}

The remainder of this section is devoted to proving Theorem~\ref{Big}.

\subsection{The algorithm: Handling insertion or deletion of an edge}
We now show how to maintain the invariant under edge insertions and deletions.
A vertex is called \emph{dirty} if it violates Invariant~\ref{inv}, and \emph{clean} otherwise.
Initially, the graph is empty, so that every vertex is clean and is at level zero.
Assume that at the time instant just prior to the $t^{th}$ update, all vertices are clean.
When the $t^{th}$ update takes place, which either inserts or deletes an edge $e = (u, v)$,
we need to adjust the weights of $u$ and $v$ accordingly.
Due to this adjustment, the vertices $u$, or $v$, or both may become dirty.
To recover from this, we call the procedure \textsc{Fix}.
The pseudo codes of the update algorithm (Algorithm~1) and the procedure \textsc{Fix}
are shown in the next page.

\alglanguage{pseudocode}
\algrenewcommand\algorithmicrequire{\textbf{Input}}
\algrenewcommand\algorithmicensure{\textbf{Output}}

\begin{algorithm}
\begin{algorithmic}[1]
\caption{}
\If {an edge $e = (u,v)$ has been inserted }
\State Set $\ell(e) = \max\,\{\ell(u),\ell(v)\}$ and set $w(u,v)= \mu\beta^{-\ell(e)}$
\State Update $W_u$  and $W_v$

\ElsIf {an edge $e = (u,v)$ has been deleted }
\State Update $W_u$  and $W_v$
\EndIf

\State Run procedure \textsc{Fix}
\end{algorithmic}
\end{algorithm}

\begin{algorithm*}
\normalsize{\bf procedure} {\textsc{Fix}}:
\begin{algorithmic}[1]

\While {there exists a dirty vertex $v$}

\If {$W_{v}> c_v$}
\State Increment the level of $v$ by setting $\ell(v) \leftarrow \ell(v)+1$
\State Update $W_v$ and $W_u$ for all affected $v$'s neighboring vertices $u$
\ElsIf{$W_{v} < c_v^\ast$ and $\ell(v) > 0$}
\State Decrement the level of $v$ by setting $\ell(v) \leftarrow \ell(v) - 1$
\State Update $W_v$ and $W_u$ for all affected $v$'s neighboring vertices $u$
\EndIf

\EndWhile
\end{algorithmic}
\end{algorithm*}

Algorithm~1 ensures that Invariant~\ref{inv} is maintained after each update, so that the dynamic scheme is $\alpha(\beta+1)$-tight as desired.
To complete the discussion, as well as the proof of Theorem~\ref{Big}, it remains to show that each update can be performed efficiently, in amortized $O(\log n)$ time.

\longdelete{
\subsection{Data structures}
In this section we introduce the data structures to implement.
Here, we use the same structure reported in~\cite{BHI2}. \\
- A counter LEVEL$[v]$ to keep track of the current level of $v$. Thus,
we set LEVEL$[v] \leftarrow \ell(v)$.\\
- A counter WEIGHT$[v]$ to keep track of the weight of $v$. Thus, we set WEIGHT$[v] \leftarrow W_{v}$.\\
- For every level $i > $ LEVEL$[v]$, we store the set of vertices $N_{v}(i)$ as a doubly linked list
NEIGHBORS$_{v}[i]$. \\
- For every level $i\leq$ LEVEL$[v]$, the list NEIGHBORS$_{v}[i]$ is empty.\\
- For level $i$ = LEVEL$[v]$, we store the set of edges $N_{v}(0,i)$ in the form of
a doubly linked list NEIGHBORS$_{v}[0, i]$.\\
- For every level $i =$ LEVEL$[v]$, the list NEIGHBORS$_{v}(0,i)$ is empty.
In addition,
we have a doubly linked list DIRTY-NODES to store the set of dirty vertices.
For every vertex $v \in V$,
we maintain a flag bit STATUS$[v] \in \{$\emph{dirty, clean}$\}$
that indicates whether the vertex $v$ is dirty or clean.
}

\subsection{Time complexity}
Each update involves two steps, namely the adjustment of weights of the endpoints, and the running of procedure \textsc{Fix}.
We now give the time complexity analysis, where the main idea is to prove the following two facts: {\bf (Fact~1)} the amortized cost of the adjustment step is $O(\log n)$, and
{\bf (Fact~2)} the amortized cost of the procedure \textsc{Fix} is zero, irrespective of the number of vertices or edges that are affected during this step.
Once the above two facts are proven, the time complexity analysis follows.

\longdelete
{
\begin{lemma}\label{Time bound}
Fix any $\epsilon$, $0 < \epsilon <0.1$,
$\alpha = (2\beta-1)/(\beta-1)$
and
$\beta = 2.84 $.
Initializing
an empty graph $G=(V, \emptyset)$, we can maintain an $\alpha(\beta+1)$-tight level scheme in $G$
which satisfies Invariant~\ref{inv} in $O(\log n / \epsilon)$ amortized update time.
\end{lemma}
}

We use the standard potential method in our amortized analysis.
Imagine that we have a bank account $B$. Initially, the graph is empty, and the bank account $B$ has no money.
For each adjustment step during an edge insertion or deletion, we deposit some money into the bank account $B$;
after that, we use the money in $B$ to pay for the cost of the procedure \textsc{Fix}.

Following the definition of~\cite{BHI2}, we say a vertex $v \in V$ is \emph{active} if its degree in $G$ is non-zero, and \emph{passive} otherwise.
Now, the value of $B$ is set by the following formula:
\begin{align*}
B\ =\ \frac{1}{\epsilon}\, \cdot\, \left(\, \sum_{e\in E}{\phi(e)}+\sum_{v\in V}{\psi(v)}\, \right),
\end{align*}
where $0 < \epsilon < 1$, and $\phi$ and $\psi$ are functions defined as follows:

\begin{align*}
\phi(e) = \left(\frac{\beta}{(\beta-1)}+\epsilon\right)(L-\ell(e)).
\end{align*}

\begin{align*}
\psi(v)= \begin{dcases}
             \ \frac{\beta^{(\ell(v)+1)}}{\mu(\beta-1)}\cdot \max\, \{0,\alpha\, c_v^\ast-W_{v} \}, & \text{if $v$ is \emph{active}}. \\
             \ 0, & \text{otherwise.}
           \end{dcases}
\end{align*}

The following lemma proves Fact~1.

\begin{lemma}\label{Total potential}
After the adjustment step, the potential $B$ increases by at most $O(\log n/\epsilon)$.
\end{lemma}

\begin{proof}
We separate the discussion into two cases: edge insertion and edge deletion.
Let $t$ be the moment where the update occurs.

\begin{itemize}
\item
{\bf Edge insertion.} The inserted edge $e$
generates a change of at most $\left(\frac{\beta}{(\beta-1)}+\epsilon\right)L$ in $\phi(e)$.  So, the summation $\sum \phi(e)$ increases by
at most $O(\log n)$.  For each endpoint $v$ of $e$, there are two possible cases for the change in $\psi(v)$:
\begin{description}
\item[Case 1:] The vertex $v$ was \emph{passive} at moment $t-1$.
By the definition of $\psi(v)$, we had $\psi(v) = 0$ and $\ell(v) = 0$ before the insertion of the edge $e$.
Hence, after the insertion of $e$,
we have
$$\psi(v)\ =\ \frac{\beta}{\mu(\beta-1)}\cdot \max\, \{0,\alpha\, c_v^\ast-W_{v}\}\ \leq\ \frac{\beta}{\mu(\beta-1)}\cdot{\alpha\, c_v^\ast}\ \leq\ \frac{\beta}{\mu(\beta-1)}\cdot{c_v}\ <\ \frac{\beta}{\beta-1}.$$
Therefore, the summation $\sum \psi(v)$ increases by at most $O(1)$.
\item[Case 2:] The vertex $v$ was \emph{active} at moment $t-1$.
In this case, the vertex $v$ remains \emph{active} at moment $t$.
Thus, the weight $W_v$ increases, and $\psi(v)$ can only decrease.
\end{description}
In both cases, the total potential $B$ increases by at most $O(\log n/\epsilon)$ after an edge insertion.

\item
{\bf Edge deletion.}
If an edge $e$ is deleted from $E$,
then $\phi(e)$ drops to zero, so that the summation $\sum \phi(e)$ decreases.
In contrast, the weight $W_{v}$ of each endpoint $v$ of $e$ decreases by at most~$\mu \beta^{-\ell(v)}$.
So,
$\psi(v)$ increases by at most
$$ \frac{\beta^{(\ell(v)+1)}}{ \mu(\beta-1)}\ \cdot\ \mu \beta^{-\ell(v)}
= \frac{\beta}{\beta-1},$$
which is a constant.  Thus, the summation $\sum \psi(v)$ increases by at most $O(1)$.  In summary,
the total potential $B$ increases by at most $O(1/\epsilon)$ after an edge deletion.
\end{itemize}

By the above arguments, the lemma follows.
\end{proof}


We now switch our attention to Fact~2.  Observe that the procedure \textsc{Fix} performs a series of level up and level down events.
For each such event, the level of a specific vertex~$v$ will be changed, which will then incur a change in its weight,
and changes in the weights of some of the incident edges and their endpoints.  Let $t_0$ denote the moment before a level up or a level down event, and $t_1$ denote the moment after the weights of the edges and vertices are updated due to this event.
Let $\textsc{Count}$ denote the number of times an edge in the graph $G$ is updated (for simplicity, we assume that in one edge update, the weight and the assignment of the edge may be updated, and so do the weights of its endpoints, where all these can be done in $O(1)$ time).

For ease of notation, in the following, a superscript $t$ in a variable denotes the variable at moment $t$.  For instance, $W_v^{t_0}$ stands for the weight $W_v$
of $v$ at moment $t_0$.  Also, we use $\Delta x$ to denote
the quantity $x^{t_0} - x^{t_1}$, so that
$$|\Delta \textsc{Count}| = |\textsc{Count}^{t_0} - \textsc{Count}^{t_1}| = \textsc{Count}^{t_1} - \textsc{Count}^{t_0}$$ represents the number of incident edges whose weights are changed between $t_0$ and $t_1$.

\medskip

Briefly speaking, based on the level scheme and the potential function $B$, we can show:
\begin{itemize}
  \item For each level up event, each of the affected edges $e$ would have its $\phi(e)$ value dropped,
          so that an $\epsilon$ fraction can pay for the weight updates of itself and its endpoints, while the remaining fraction
          can be converted into the increase in $\psi(v)$ value.

  \item For each level down event, the reverse happens, where the vertex $v$ would have its $\psi(v)$ value dropped,
          so that an $\epsilon$ fraction can pay for the weight updates of the affected edges and their endpoints,
          while the remaining fraction can be converted into the increase in $\phi(e)$ values of the affected edges.
          The $\alpha$ value controls the frequency of the level down events, while trading this off with the approximation guarantee.
\end{itemize}

Sections~\ref{level up} and~\ref{level down} present the details of the amortized analysis of these two types of events, respectively.
Finally, note that there is no money (potential) input to the bank $B$ after the adjustment step,
so that the analysis implies that the procedure \textsc{Fix} must stop (as the money in the bank is finite).

\subsubsection{Amortized cost of level up} \label{level up}
Let $v$ be the vertex that undergoes the level up event, and $i = \ell(v)$ denote its level at moment $t_0$.
By our notation, $\Delta B = B^{t_0} - B^{t_1}$ denotes the potential \emph{drop} in the bank $B$ from moment $t_0$ to moment $t_1$.
To show that the amortized cost of a level up event is at most zero, it is equivalent to show that $\Delta B \geq |\Delta \textsc{Count}|$.

Recall that after a level up event, only the value of $\psi(v)$, the values of $\phi(e)$ and $\psi(u)$ for an edge $(u,v)$ may be affected.
In the following, we will examine carefully the changes in such values, and derive the desired bound for $\Delta B$.
First, we have the following simple lemma.

\begin{lemma} \label{TB1}
$|\Delta \textsc{Count}| \leq D_{v}^{t_0}(0,i)$.
\end{lemma}
\begin{proof}
When $v$ changes from level $i$ to $i+1$, only those incident edges with levels $i$
will be affected.
\end{proof}

The next three lemmas examine, respectively, the changes $\Delta \psi(v)$, $\Delta \phi(e)$, and $\Delta \psi(u)$.

\begin{lemma} \label{pv1}
$\Delta\psi(v) = 0$.
\end{lemma}
\begin{proof}
Since $v$ undergoes a level up event, we have $W_v^{t_0}
> c_v > \alpha\,c^\ast_v$, so that $$ \psi^{t_0}(v) = 0.$$

Next, we look at $\psi^{t_1}(v)$.  To begin with, we show a general relationship between $W_v^{t_0}$ and $W_v^{t_1}$,
similar to that in the proof of Lemma~\ref{non-improvable}.
Let $i$ denote the level $\ell(v)$ of $v$ at moment $t_0$.

\begin{description}
  \item[Case 1:] $D_{v}^{t_0}(0,i) > k_{v}$
  \begin{align*}
  W_{v}^{t_0} &=  W_{v}(i) = k_{v}\cdot\mu\beta^{-i}
    + \min\{k_v ,D_v^{t_0}(i+1)\}\cdot\mu\beta^{-(i+1)}+\sum_{j>i+1}{\min\{k_v,D_v^{t_0}(j)\}\mu\beta^{-j}} \\
  W_{v}^{t_1} &=W_{v}(i+1)=k_{v}\cdot\mu\beta^{-(i+1)}+\sum_{j>i+1}{\min\{k_v,D_v^{t_0}(j)\}\mu\beta^{-j}}\\
  &\geq \frac{1}{\beta+1}W_{v}^{t_0}
  \end{align*}
  \item[Case 2:] $D_{v}^{t_0}(0,i) \leq k_{v}$
  \begin{align*}
  W_{v}^{t_0} &=  W_{v}(i) = D_{v}^{t_0}(0,i)\cdot\mu\beta^{-i}+\min\{k_v ,D_v^{t_0}(i+1)\}\cdot\mu\beta^{-(i+1)}+\sum_{j>i+1}{\min\{k_v,D_v^{t_0}(j)\}\mu\beta^{-j}} \\
  W_{v}^{t_1} &= W_{v}(i+1)=\min\{k_v,D_v^{t_0}(0,i+1)\}\cdot\mu\beta^{-(i+1)}+\sum_{j>i+1}{\min\{k_v,D_v^{t_0}(j)\}\mu\beta^{-j}}\\
  &\geq \frac{1}{\beta+1}W_{v}^{t_0}
  \end{align*}
\end{description}
Thus, $W_v^{t_1} \geq W_v^{t_0}/(\beta+1) > c_v/(\beta+1) = \alpha\,c^\ast_v$, which implies $$\psi^{t_1}(v) = 0.$$
In summary, we have $\Delta \psi(v) = 0 - 0 = 0$ as desired.
\end{proof}

\begin{lemma}
For every edge $e$ incident to $v$,
\begin{equation*}
  \Delta\phi(e)=\begin{dcases}
                \ \left(\frac{\beta}{(\beta-1)}+\epsilon\right), & \text{if $\ell(e)\in[0,i]$.}\\
                \ 0, & \text{otherwise}.
              \end{dcases}
\end{equation*}
\end{lemma}\label{pe1}
\begin{proof}
As mentioned, only those edges that are at the level in the range $[0,i]$ are affected, so that
\begin{align*}
 \Delta\phi(u,v)
 &=\phi^{t_0}(u,v)-\phi^{t_1}(u,v) \\
 &=
 \left(\frac{\beta}{(\beta-1)}+\epsilon\right)(L-i)-\left(\frac{\beta}{(\beta-1)}+\epsilon\right)(L-(i+1))=\left(\frac{\beta}{(\beta-1)}+\epsilon\right).
\end{align*}
\end{proof}

\begin{lemma}
For every vertex $u \in N_{v}^{t_0}$, $\Delta\psi(u) \geq - \beta / (\beta-1)$.
\end{lemma}\label{pu1}
\begin{proof}
If $\ell(u)\in[i+1,L]$, then $w^{t_0}(u,v)=w^{t_1}(u,v)$ and thus $\Delta w(u,v)=0$,
which implies that $\Delta\psi(u)=0$.
The potential $\psi(u)$ changes only when the vertex $u$ lies at the level in the range $[0,i]$.
Without loss of generality, we assume $\ell(u) = i$ and
prove the lemma by considering the relationship between
$k_{u}$, $D_{u}^{t_0}(0,i)$ and $D_{u}^{t_0}(i+1)$.
For those vertices $u$ with $\ell(u)\in[0,i-1]$,
we replace the term $D_{u}^{t_0}(0,i)$ to $D_{u}^{t_0}(i)$ still achieve the same result.

\begin{description}
   \item[Case 1:]$D_{u}^{t_0}(0,i)> k_{u}$, $D_{u}^{t_0}(i+1)\geq k_{u}$
   \begin{align*}
   W_u^{t_0}\ =\ W_u^{t_1}\quad\quad \Rightarrow \quad\quad \Delta\psi(u)\ =\ 0.
   \end{align*}

 \item[Case 2:]$D_{u}^{t_0}(0,i)> k_{u}$, $D_{u}^{t_0}(i+1)< k_{u}$
 \begin{align*}
 W_u^{t_1} &\ =\  W_u^{t_0}+\mu\beta^{-(i+1)} \\
\Delta\psi(u)
&\ =\ \frac{\beta^{(\ell(u)+1)}}{\mu(\beta-1)}\cdot\mu\beta^{-(i+1)}\ =\  \frac{1}{\beta-1}\cdot\beta^{\ell(u)-i}\ >\ 0.
\end{align*}


\item[Case 3:]$D_{u}^{t_0}(0,i)\leq k_{u}$, $D_{u}^{t_0}(i+1)< k_{u}$
\begin{align*}
W_u^{t_1}&\ =\ W_u^{t_0}-\mu(\beta^{-i}-\beta^{-(i+1)}) \\
\Delta\psi(u)
&\ =\ -\frac{\beta^{(\ell(u)+1)}}{\mu(\beta-1)}\cdot\mu(\beta^{-i}-\beta^{-(i+1)})\ =\ -\frac{\beta^{(\ell(u)+1)}}{\beta^{i+1}}\ \geq\ -1.
\end{align*}

\item[Case 4:]$D_{u}^{t_0}(0,i)\leq k_{u}$, $D_{u}^{t_0}(i+1)\geq k_{u}$
\begin{align*}
W_u^{t_1}&\ = W_u^{t_0}-\mu\beta^{-i} \\
\Delta\psi(u)
&\ =\ -\frac{\beta^{(\ell(u)+1)}}{\mu(\beta-1)}\cdot\mu\beta^{-i}\ =\ -\frac{\beta}{\beta-1}(\beta^{\ell(u)-i})\ \geq\ -\frac{\beta}{\beta-1}.
\end{align*}
\end{description}
\end{proof}

Based on the above lemmas, we derive the following and finish the proof for the case of level up.
\begin{align*}
\Delta B &\ =\ \frac{1}{\epsilon}\, \cdot\left(\Delta\psi(v)+\sum_{e\in E}{\Delta\phi(e)}+\sum_{u\in N_{v}^{t_0}}{\Delta\psi(u)}\right)\\
&\ \geq\ \frac{1}{\epsilon}\, \cdot\left(0+\left(\frac{\beta}{(\beta-1)}+\epsilon\right)D_{v}^{t_0}(0,i)-\frac{\beta}{\beta-1}D_{v}^{t_0}(0,i)\right)\\
&\ =\  D_{v}^{t_0}(0,i)\ \ \geq\ \ |\Delta \textsc{Count}|.
\end{align*}


\subsubsection{Amortized cost of level down} \label{level down}
We now show that the amortized cost of level down for a vertex $v$ is at most zero.
Similar to the case of level up, we examine $\Delta \psi(v)$, $\Delta \phi(e)$, and $\Delta \psi(u)$,
and show that $\Delta B \geq |\Delta \textsc{Count}|$.

Before starting the proof of the level down case, recall that we have mentioned a parameter~$h$ at the end of Introduction,
where $h$ is the largest number of selected copies of all the vertices.  That is, $ h = \max_v \{ \lceil |\delta^{t_0}(v)|/k_v \rceil \}$.
Also, we let $h' = \max_v \{ \lceil D_v^{t_0}(0,\ell(v))/k_v \rceil \}$, where $h' \geq h$, and set $\xi \geq 0$ such that $h' = h + \xi$.

\begin{lemma}\label{TB2}
$|\Delta \textsc{Count}| \leq D_{v}^{t_0}(0,i) < h' \cdot\frac{\beta^{i}c_v^\ast}{\mu}$.
\end{lemma}
\begin{proof}
When the vertex $v$ moves from level $i$ to $i-1$, only those edges whose levels are at most $i$ are affected.
This shows the first part of the inequality.
Also, because $v$ undergoes level  down, we have $W_v^{t_0} < c^\ast_v$.
Then, for the latter inequality, we partition the proof into two cases:
\begin{description}
  \item[Case 1:] $\lceil|\delta^{t_0}(v)|/k_v\rceil = 1$
 \begin{align*}
 & \ W_{v}^{t_0}= D_{v}^{t_0}(0,i)\cdot\mu\beta^{-i}+\sum_{j>i}{\min\{k_v,D_v^{t_0}(j)\}\mu\beta^{-j}}\\
\Rightarrow\ &\ c_v^\ast\ >\ D_{v}^{t_0}(0,i)\cdot\mu\beta^{-i}\\
\Rightarrow\ &\ D_{v}^{t_0}(0,i)\ < \ \frac{\beta^{i}c_v^\ast}{\mu}.
  \end{align*}
  \item[Case 2:] $\lceil|\delta^{t_0}(v)|/k_v\rceil > 1  $
  \begin{align*}
  & \ W_{v}^{t_0}= k_{v}\cdot\mu\beta^{-i}+\sum_{j>i}{\min\{k_v,D_v^{t_0}(j)\}\mu\beta^{-j}}\\
\Rightarrow\ & \ c_v^\ast > k_{v}\cdot\mu\beta^{-i}\\
\Rightarrow\ & \ h' c_v^\ast > \frac{D_{v}^{t_0}(0,i)}{k_{v}}\cdot k_{v}\cdot\mu\beta^{-i}\\
\Rightarrow\ & \ D_{v}^{t_0}(0,i) < h' \cdot\frac{\beta^{i}c_v^\ast}{\mu}.
  \end{align*}
\end{description}
\end{proof}

Now, we are ready to examine $\Delta \psi(v)$, $\Delta \phi(e)$, and $\Delta \psi(u)$, through the following lemmas.

\begin{lemma}\label{pu2}
For every vertex $u\in N_{v}^{t_0}$, $\Delta\psi(u)\geq -1 /(\beta-1)$.
\end{lemma}
\begin{proof}
If $\ell(u)\in[i,L]$, then $w^{t_0}(u,v)=w^{t_1}(u,v)$ and $\Delta w(u,v)=0$,
which implies $\Delta\psi(u)=0$.
The changes of potentials only occur at the vertex whose level is in the range $[0,i-1]$.
WLOG, we assume $\ell(u) = i-1$ and
we consider the relationship between
$k_{u}$, $D_{u}^{t_0}(0,i-1)$ and $D_{u}^{t_0}(i)$.
For those vertices $u$ and $\ell(u)\in[0,i-2]$,
we replace the notion $D_{u}^{t_0}(0,i-1)$ to $D_{u}^{t_0}(i-1)$ still having the same result.
  \begin{description}
    \item[Case 1:] $D_{u}^{t_0}(0,i-1)\geq k_{u}$, $D_{u}^{t_0}(i)> k_{u}$
    \begin{align*}
    W_u^{t_0}\ =\ W_u^{t_1}\quad\quad \Rightarrow\quad\quad \Delta\psi(u)\ =\ 0.
    \end{align*}
    \item[Case 2:] $D_{u}^{t_0}(0,i-1)\geq k_{u}$, $D_{u}^{t_0}(i)\leq k_{u}$
    \begin{align*}
    W_u^{t_1}
    &\ =\ W_u^{t_0}-\mu\beta^{-i}\\
    \Delta\psi(u)
    &\ =\ -\frac{\beta^{(\ell(u)+1)}}{\mu(\beta-1)}\cdot\mu\beta^{-i}\\
    &\ =\ -\frac{\beta}{\beta-1}\cdot\beta^{\ell(u)-i}
    \hspace*{0.5cm}(\because \ell(u)\leq i-1)\\
    &\ \geq \ -\frac{1}{\beta-1}.
    \end{align*}
    \item[Case 3:]$D_{u}^{t_0}(0,i-1)< k_{u}$, $D_{u}^{t_0}(i)> k_{u}$
    \begin{align*}
    &\ W_u^{t_1}= W_u^{t_0}+\mu\beta^{-(i-1)}\\
    \Rightarrow\ &\ \psi^{t_0}(u)\ >\ \psi^{t_1}(u)\quad\quad\quad \Rightarrow\quad\quad \Delta\psi(u)\ >\ 0.\\
    \end{align*}
    \item[Case 4:]$D_{u}^{t_0}(0,i-1)< k_{u}$, $D_{u}^{t_0}(i)\leq k_{u}$
    \begin{align*}
    &W_u^{t_1}= W_u^{t_0}+\mu(\beta^{-(i-1)}-\beta^{-i})\\
    \Rightarrow\ &\ \psi^{t_0}(u)\ >\ \psi^{t_1}(u)\quad\quad\quad \Rightarrow\quad\quad \Delta\psi(u)\ >\ 0.\\
    \end{align*}
  \end{description}
\end{proof}
Next, we partition $N_{v}^{t_0}$ into three subsets: $X$, $Y_1$ and $Y_2$,
i.e. $N_{v}^{t_0}=X\cup Y_1\cup Y_2$, where
\begin{align*}
X & = \{ u \mid u \in N_{v}^{t_0}(0,i-1) \}, \\
Y_1 & = \{ u \mid u \in N_{v}^{t_0}(i) \}, \\
Y_2 & = \{ u \mid u \in N_{v}^{t_0}(i+1,L) \}.
\end{align*}
\longdelete{ 
\begin{equation*}
  \begin{cases}
    X, & \text{if $u\in N_{v}^{t_0}(0,k-1)$ } \\
    Y_1, & \text{if $u\in N_{v}^{t_0}(k)$} \\
    Y_2, & \text{if $u\in N_{v}^{t_0}(k+1,L)$}
  \end{cases}
\end{equation*}
}


\begin{lemma}\label{pe2}
For every edge $(u,v)$ incidents to a vertex $v$,
\begin{equation*}
  \Delta\phi(u,v)=\begin{dcases}
                \ -\left(\frac{\beta}{(\beta-1)}+\epsilon\right), & \text{if $u\in X$}\\
                \ 0, & \text{if $u\in Y_1\cup Y_2$}.
              \end{dcases}
\end{equation*}
\end{lemma}
\begin{proof}
Fix any vertex $u\in N_{v}^{t_0}$. We consider the following two possible scenarios.
\begin{description}
\item[Case 1:]$u\in Y_1 \cup Y_2$

When the level of the vertex $v$ decreases from $i$ to $i-1$,
$\ell^{t_0}(u,v)= \ell^{t_1}(u,v)$ and
thus $\phi^{t_0}(u,v) = \phi^{t_1}(u,v)$,
which implies
$\Delta\phi(u,v)=0$.

\item[Case 2:]$u\in X$

When the level of the vertex $v$ decreases from $i$ to $i-1$,
we
have
$\ell^{t_0}(u,v) = i$ and $\ell^{t_1}(u,v) = (i-1)$.
The following result is thus derived:
\begin{align*}
\Delta\phi(u,v) = \left(\frac{\beta}{(\beta-1)}+\epsilon\right)(L-i)- \left(\frac{\beta}{(\beta-1)}+\epsilon\right)(L-i+1)= -\left(\frac{\beta}{(\beta-1)}+\epsilon\right).
\end{align*}

\end{description}
\end{proof}

\medskip

Next, let $W_v^{t_0}= x + y_1 + y_2$, where $x$, $y_1$ and $y_2$ on the right-hand-side
correspond to the weights generated by the subsets $X$, $Y_1$, $Y_2$, respectively.
So, we get the following lemmas:

\begin{lemma}\label{Tpe2}
$\sum_{u\in N_{v}^{t_0}}{\Delta\phi(u,v)}\leq
-\left(\frac{\beta}{(\beta-1)}+\epsilon\right)( hx\beta^{i} / \mu ).$
\end{lemma}


\begin{proof}
We consider $|X|$ in the following two cases:
\begin{description}
  \item[Case 1:]$|X|\ \leq \ k_{v}$
  \begin{align*}
  x = |X|\cdot\mu\beta^{-i}\quad\quad \Rightarrow\quad\quad |X|=\frac{x\beta^{i}}{\mu}\leq{hx\beta^{i}}/{\mu}
  \hspace*{0.5cm}\text{($\because$ $h \ge 1$ )}
  \end{align*}
  \item[Case 2:]$|X|\ > \ k_{v}$.  Here, we may assume, WLOG, that $y_1 = 0$.  Then, we have
  \begin{align*}
   &\ x = k_{v}\cdot\mu\beta^{-i} \\
\Rightarrow\ &\ |X|=\frac{|X|}{k_{v}}\cdot\frac{x\beta^{i}}{\mu}\leq \left\lceil\frac{\delta^{t_0}(v)}{k_{v}} \right\rceil\cdot\frac{x\beta^{i}}{\mu}
\leq h\cdot\frac{x\beta^{i}}{\mu}.
  \end{align*}
\end{description}
Finally, since $$\sum_{u\in N_{v}^{t_0}}{\Delta\phi(u,v)}=|X|\cdot-\left(\frac{\beta}{(\beta-1)}+\epsilon\right),$$
the lemma thus follows.
\end{proof}

\begin{lemma}\label{pv2}
$\Delta\psi(v)=(\alpha c_v^\ast-x-y_1-y_2)\cdot\frac{\beta^{i+1}}{\mu(\beta-1)}-\max \{ 0,\alpha c_v^\ast-\beta x-y_1-y_2 \} \cdot\frac{\beta^{i}}{\mu(\beta-1)}$.
\end{lemma}
\begin{proof}
We have
$W_{v}^{t_0}=x+y_1+y_2< c_v^\ast$, and
we have to consider the following relationship between $x+y_1$ and $k_{v}\cdot\mu\beta^{-i}$.
With the above relationship, we compute $W_v^{t_1}$ by the following:
\begin{description}
\item[Case 1] $|X|<k_{v}$ and $|X+Y_1|\leq k_{v}$:
\begin{align*}
W_{v}^{t_0}=W_{v}(i) &=  D_{v}^{t_0}(0,i)\mu\beta^{-(i)}+\sum_{j>i}{\min\{k_v,D_v^{t_0}(j)\}\mu\beta^{-j}};\\
W_{v}^{t_1}=W_{v}(i-1) &= D_{v}^{t_0}(0,i-1)\mu\beta^{-(i-1)}+D_{v}^{t_0}(i)\mu\beta^{-(i)}+\sum_{j>i+1}{\min\{k_v,D_v^{t_0}(j)\}\mu\beta^{-j}}\\
&= \beta x+y_1+y_2;
\end{align*}
\item[Case 2] $|X+Y_1| > k_{v}$:
\begin{align*}
W_{v}^{t_0}=W_{v}(i) &= k_{v}\mu\beta^{-(i)} + \sum_{j>i}{\min\{k_v,D_v^{t_0}(j)\}\mu\beta^{-j}}; \\
W_{v}^{t_1}=W_{v}(i-1) &= \min\{k_v,D_{v}^{t_0}(0,i-1)\}\mu\beta^{-(i-1)}+\min\{k_v,D_{v}^{t_0}(i)\}\mu\beta^{-(i)}+\sum_{j>i}{\min\{k_v,D_v^{t_0}(j)\}\mu\beta^{-j}}\\
&\leq (\beta+1) x+y_1+y_2;
\end{align*}
\end{description}
By the above cases, we have a weight change of at least $\beta x+y_1+y_2$ in $W_v$.
The desired bound on $\Delta \psi(v)$ can thus be obtained by direct substitution.
\end{proof}

Finally, depending upon the value of $\alpha c_v^\ast-\beta x-y_1-y_2$,
we consider two possible scenarios,
where we show that in each case, $\Delta B\geq h' \cdot\beta^{i}c_v^\ast / \mu$. This in turn implies $\Delta B \geq |\Delta \textsc{Count}|$
as desired.
\begin{description}
  \item[Case 1:]$\alpha c_v^\ast\leq\beta x+y_1+y_2$
  \begin{align*}
\epsilon\cdot\Delta B&=\left(\sum_{u\in N_{v}^{t_0}}{\Delta\psi(u)}+\sum_{e\in E}{\Delta\phi(e)}+\Delta\psi(v)\right)\\
  &\geq-\frac{1}{\beta-1}\cdot \frac{hx\beta^{i}}{\mu}-\left(\frac{\beta}{(\beta-1)}+\epsilon\right)\cdot\frac{hx\beta^{i}}{\mu}+(\alpha c_v^\ast- x-y_1-y_2)\cdot\frac{\beta^{i+1}}{\mu(\beta-1)}\\
  &\geq\frac{\beta^i}{\mu}\left(-\frac{1}{\beta-1}hc_v^\ast-\left(\frac{\beta}{(\beta-1)}+\epsilon\right)hc_v^\ast+\frac{(\alpha-1)\beta}{\beta-1}c_v^\ast\right)
  \hspace*{0.5cm}\text{($\because c_v^\ast \geq x+y_1+y_2 $)}\\
  &=\frac{\beta^{i}c_v^\ast}{\mu(\beta-1)}((\alpha-1)\beta-h-(\beta+(\beta-1)\epsilon)h)\\
  &=\frac{\beta^{i}c_v^\ast}{\mu}\left(\frac{(\alpha-1)\beta}{(\beta-1)}-h\left(\frac{\beta+1}{\beta-1}+\epsilon\right)\right)
  \hspace*{0.5cm}\text{if let $\alpha = \frac{\beta-1}{\beta}\left(h\left(\frac{\beta+1}{\beta-1}+2\epsilon\right)+\xi\epsilon\right)+1$}\\
  &\geq \epsilon h'\cdot\frac{\beta^{i}c_v^\ast}{\mu}.
  \end{align*}
  \item[Case 2:]$\alpha c_v^\ast>\beta x+y_1+y_2$
  \begin{align*}
  \epsilon\cdot\Delta B &=\left(\sum_{u\in N_{v}^{t_0}}{\Delta\psi(u)}+\sum_{e\in E}{\Delta\phi(e)}+\Delta\psi(v)\right)\\
  &\geq-\frac{1}{\beta-1}\cdot\frac{hx\beta^{i}}{\mu}-\left(\frac{\beta}{(\beta-1)}+\epsilon\right)\cdot\frac{hx\beta^{i}}{\mu}+(\alpha c_v^\ast- x-y_1-y_2)\cdot\frac{\beta^{i+1}}{\mu(\beta-1)}\\
  &-(\alpha c_v^\ast-\beta x-y_1-y_2)\cdot\frac{\beta^{i}}{\mu(\beta-1)}\\
  &=\frac{\beta^{i}}{\mu(\beta-1)}\cdot(-xh-(\beta+(\beta-1)\epsilon)xh+\alpha(\beta-1)c_v^\ast-(\beta-1)(y_1+y_2))\\
  &=\frac{\beta^{i}}{\mu(\beta-1)}\cdot(\alpha(\beta-1)c_v^\ast-(\beta+1+(\beta-1)\epsilon)xh-(\beta-1)(y_1+y_2))\\
  &\geq\frac{\beta^{i}c_v^\ast}{\mu}\cdot\left(\alpha-h\left(\frac{\beta+1}{\beta-1}+\epsilon\right)\right)
  \hspace*{0.5cm}\text{if let $\alpha = \frac{\beta-1}{\beta}\left(h\left(\frac{\beta+1}{\beta-1}+2\epsilon\right)+\xi\epsilon\right)+1$}\\
  &\geq\epsilon h'\cdot\frac{\beta^{i}c_v^\ast}{\mu}.
  \end{align*}
\end{description}
Thus, the level scheme remains $\alpha(\beta+1)$-tight after a level down event. However, the value of $h$ is bounded by $n$,
and $h$ appears inside $\alpha$, so that the approximation ratio of the scheme may become $n$ in the worst-case.
Fortunately, with the help of the following lemma, we can choose $\alpha$ carefully, which in turn
improves the approximation ratio from $n$ to $O(1)$.

\begin{lemma}\label{non-dec}
Suppose that we set $\alpha \geq \beta/(\beta-1)$.  By the time a level down event occurs at $v$ at moment $t_0$,
exactly one copy of $v$ is selected.  That is, $\lceil |\delta^{t_0}(v)| / k_v \rceil = 1$.
\end{lemma}

\begin{proof}
Assume to the contrary that $v$ could decrease its level even if more than one copy of $v$ is selected.
Since $v$ undergoes level down, its weight $W_v$ must have decreased; this can happen only in one of the following cases:

\begin{description}
  \item[Case 1:] An incident edge whose level is in the range $[0,\ell(v)]$ is deleted.  In this case,
  since more than one copy of $v$ is selected, $W_v$ is unchanged.  Thus, this case cannot happen.

  \item[Case 2:] An incident edge whose level is in the range $[\ell(v)+1,L]$ is deleted.
  In this case, the weight $W_v^{t_0}$ at moment $t_0$ is less than $c^\ast_v$.
  On the other hand, at the moment $t'$ when $v$ attains the current level $\ell(v)$ (from level $\ell(v)-1$),
  its weight $W_v^{t'}$ was at least $c_v$ before level up, and became at least $c_v/(\beta+1)$ after the level up.
  (The reason is from the proof of Lemma~\ref{non-improvable}:  the weight change between consecutive levels is at most a factor of $\beta+1$.)
  This implies that:
  \begin{align*}
  c^\ast_v\ &\ >\ W_v^{t_0}\ \geq\ k_v \mu\beta^{-\ell(v)}  \hspace{0.5cm}\because\ \mbox{more than one copy of $v$ is selected}\\
  (\beta/(\beta-1)) k_v \mu\beta^{-\ell(v)} &\ \geq W_v^{t'}\ \geq\ c_v/(\beta+1)  \hspace{0.5cm}\because\ \mbox{left bound is max possible $W_v$ value}
  \end{align*}
  Combining, we would have
  $$ \frac{c_v}{\alpha(\beta+1)}\ =\ c^\ast_v\ >\ k_v \mu\beta^{-\ell(v)}\ \geq\ \frac{c_v(\beta-1)}{\beta(\beta+1)},  $$
  so that $ \alpha < \beta/(\beta-1)$.  A contradiction occurs.
\end{description}
Thus, the lemma follows.
\end{proof}

The above lemma states that if we choose $\alpha \geq \beta/(\beta-1)$, then level down of $v$ occurs only when
$\lceil |\delta^{t_0}(v)|/k_v \rceil$ is one.
Then, Case 2 inside the proof of Lemma~\ref{TB2} will not occur, so that we can strengthen Lemma~\ref{TB2} to get
$|\Delta \textsc{Count}| \leq D_v^{t_0}(0,i) < \beta^i c^*_v/\mu$.
Similarly, the proof of Lemma~\ref{Tpe2} can be revised, so that we can strengthen Lemma~\ref{Tpe2} by replacing $h$ with one.
On the other hand, we need $\alpha \geq (2\beta+1)/\beta+2\epsilon$ to satisfy the amortized cost analysis.
Consequently, we set $\alpha = (2\beta+1)/\beta+2\epsilon$,
and we can achieve the desired bound $\Delta B \geq \beta^i c_v^\ast/\mu \geq |\Delta \textsc{Count}|$.
The proof for the level down case is complete.

\subsection{Summary and extensions}
With the appropriate setting of $\alpha = (2\beta+1)/\beta+2\epsilon$, where $0 < \epsilon < 1$,
we get an $\alpha(\beta+1)$-tight level scheme. Then, by setting $\beta = 2.43$, Theorem~\ref{Big} is proven
so that we get an approximation solution of ratio close to 36 with $O((\log n)/\epsilon)$ amortized update time.
Note that if we focus on the non-capacitated case, that is, each vertex is weighted and has unlimited capacity,
the problem becomes the \emph{weighted} vertex cover problem.  Our dynamic scheme can easily be adapted
to maintain an approximate solution, based on the following changes.
First, we define the weight of a vertex $W_v$ as $W_v = \sum_{e\sim v}{\mu\beta^{-\ell(e)}}$.
Next, we let $\alpha = 1+3\epsilon$ and $\beta = 1+\epsilon$
and revise $\phi(e)$ as  $\phi(e) = \left(1+\epsilon\right)(L-\ell(e))$.
\longdelete{
\begin{align*}
\psi(v)= \begin{dcases}
             \ \frac{\beta^{(\ell(v)+1)}}{\mu(\beta-1)}\cdot \max\, \{0,\alpha\, c_v^\ast-W_{v} \}, & \text{if $v$ is \emph{active}}. \\
             \ 0, & \text{otherwise.}
           \end{dcases}
\end{align*}
}
After these changes, we can go through a similar analysis, and obtain a $(2+\epsilon)$-approximate weighted vertex cover
with $O(\log n / \epsilon^2)$ amortized update time.

Finally, we consider two natural extensions of the capacitated vertex cover problem,
and show how to adapt the proposed level scheme to handle these extensions

\noindent
\textbf{Capacitated set cover.}
Here, we consider the capacitated set cover problem
which is equivalent to the capacitated vertex cover problem in hyper-graphs.
A hyper-graph $G = (V,E)$ has $|V| = n$ vertices and $|E| = m$ hyper-edges, where each hyper-edge is incident to
a set of vertices.  Suppose that each hyper-edge is incident to at most $f$ vertices.
Our target is to find a subset of vertices, each with a certain number of copies, so that every edge in $E$ is covered,
while the total cost of the selected vertices (each weighted by the corresponding number of copies) is minimized.
Here, we treat the hyper-graph vertex cover problem as if the original vertex cover problem, and
use the same level scheme and the definition of the weight of a vertex $W_v$.
That is, the weight $W_v$ of a vertex $v$ is defined as follows:
\begin{description}
  \item[Case 1] $D_v(0,\ell(v))>k_v$:
  \begin{align*}
  W_v =  k_v \mu\beta^{-\ell(v)}+\sum_{i>\ell(v)}{\min\{k_v,D_v(i)\}\mu\beta^{-i}}
  \end{align*}
  \item[Case 2] $D_v(0,\ell(v))\leq k_v$:
  \begin{align*}
  W_v = D_v(0,\ell(v))\mu\beta^{-\ell(v)}+\sum_{i>\ell(v)}{\min\{k_v,D_v(i)\}\mu\beta^{-i}}
  \end{align*}
\end{description}
We also use the same conditions for level up and level down.
However, we still need to do some adjustments for this problem.
First, we re-design the number of levels, $L$, to be $\lceil \log_{\beta}(m\mu\alpha / c_{\min})\rceil$.
Next, we adjust the flexible range by multiplying it by $f$ so that $W_v \in (c_v/f\varepsilon, c_v]$.
In Lemma~\ref{vareps-tight}, we have proved that if there are more than $k_v$ edges assigned to a vertex $v$, then every edge is accounted for at most $2(\beta/(\beta-1))\pi_e$
But here, a hyper-edge $e$ may be incident to at most $f$ vertices so that the total cost for a hyper-edge is bounded by at most $(2(\beta/(\beta-1))+(f-1))\pi_e$ instead.
Combining these with the new flexible range of $W_v$, the approximation ratio of our scheme for the capacitated set cover problem is $O(f^2)$.

When we consider the updated time in the dynamic setting,
\longdelete{
by~\cite{BHI2} the authors have proved the following argument with their proposed level scheme:
\begin{theorem}\cite{BHI2}
We can maintain an $O(f^2 + f + \epsilon f^2)$-approximation solution to the dynamic set cover
problem in $O(f\cdot \log(m+n)/\epsilon)$ amortized update time.
\end{theorem}
}
we modify our potential function as follows:
\begin{align*}
\phi(e) = \left(\frac{\beta}{(\beta-1)}+\epsilon\right)(L-\ell(e)).
\end{align*}

\begin{align*}
\psi(v)= \begin{dcases}
             \ \frac{\beta^{(\ell(v)+1)}}{f\mu(\beta-1)}\cdot \max\, \{0,f\alpha\, c_v^\ast-W_{v} \}, & \text{if $v$ is \emph{active}}. \\
             \ 0, & \text{otherwise.}
           \end{dcases}
\end{align*}
By following the arguments and proofs in Section 3 analogously,  we can prove that our scheme achieves $O(f \log(m + n)/\epsilon)$ amortized update time.

\noindent
\textbf{Capacitated vertex cover with non-uniform unsplittable demand.}
In this part, we extend the capacitated vertex cover problem to a more general case in which each edge has an unsplittable demand. That is, the demand of each edge must be covered by exactly one of its endpoints.
We first show that, with some modification, our approach in Section 2 is ready to give an $O(1)$-approximate solution for the general case in the static setting.
Firstly, when we consider the general case, we have to revise the capacity constraint in the primal problem to $k_v x_v - \sum_{e\in N_v}{y_{ev}d_e} \geq 0$, and we also have to change the vertex constraint in the dual problem
to $q_vd_e+l_{ev}\geq \pi_e$.

To cope with these changes, we will revise the number of levels of our level scheme to be $L = \lceil \log_{\beta}(k_{\rm max}\mu\alpha / c_{\min})\rceil$, where $k_{\rm max}$ denotes the maximum capacity of a vertex.
Moreover, we adjust our definition of the weight  $W_v$ of a vertex as follows:

\begin{description}
  \item[Case 1] $\sum_{e\, \mid\, e\sim v , \ell(e)=\ell(v)}{d_e} > k_v$:
  \begin{align*}
  W_v =  k_v \mu\beta^{-\ell(v)}+\sum_{j\, \mid\, \ell(e) = j >\ell(v)}{\min\{k_v,\sum_{e}{d_e}\}\mu\beta^{-j}}
  \end{align*}
  \item[Case 2] $\sum_{e\, \mid\, e\sim v , \ell(e)=\ell(v)}{d_e}\leq k_v$:
  \begin{align*}
  W_v = \sum_{e\, \mid\, e\sim v, \ell(e)=\ell(v)}{d_e} \mu\beta^{-\ell(v)}+\sum_{j\, \mid\, \ell(e) = j > \ell(v)}{\min\{k_v,\sum_{e}{d_e}\}\mu\beta^{-j}}
  \end{align*}
\end{description}
where $e \sim v$ denotes $e$ is an edge incident to $v$.

Due to the change of the mathematical model in both primal and dual problems, we need a slightly different strategy from that in Section 2.
 We use the total demand of the unassigned edges to replace the number of unassigned edges to determine the value of $q_v$ and $l_{ev}$.
In particular:

\begin{description}
    \item If $\lceil\sum_{e\in\delta(v)}{d_e}/k_v\rceil > 1$:   $q_v = \mu \beta^{-\ell(v)}$,  and $l_{ev} = 0$;
    \item If $\lceil\sum_{e\in\delta(v)}{d_e}/k_v\rceil \leq 1$:   $q_v = \mu \sum_{ i\, \mid\, \sum_{\ell(e)=i}{d_e}\, \geq\, k_v }  \beta^{-i}$,
                                                                                  $l_{ev} = 0$ if $\sum_{\ell(e)=i}{d_e}\, \geq\, k_v$, and $l_{ev} = d_e\cdot\mu \beta^{-\ell(e)}$ otherwise.
    \item For every edge $e$:  $\pi_e = d_e\cdot\mu\beta^{-\ell(e)}$.
\end{description}

Then, we use the same technique as in Section 2, and it is easy to verify that the above choices of $q_v$, $l_{ev}$, and $\pi_e$ give a feasible solution to the dual problem.
Again, for the total cost of our solution, we separate the analysis into two parts, based on the multiplicity of the vertex $v$:

\begin{description}
  \item [Case 1] $\lceil\sum_{e\in\delta(v)}{d_e}/k_v\rceil> 1$:
                             In this case, the external component of $W_v$ is at most $1/(\beta-1)$
                             of the internal component, so that $W_v \leq (\beta/(\beta-1))k_v q_v$.
                             Then, the cost of all copies of $v$ is:
  \begin{align*}
\left\lceil \sum_{e\in\delta(v)}{d_e}/k_v \right\rceil\cdot c_v &\leq \left\lceil\sum_{e\in\delta(v)}{d_e}/k_v\right\rceil\cdot \varepsilon \cdot W_v\\
    &\leq 2 \cdot \frac{ \sum_{e\in\delta(v)}{d_e}}{k_v} \cdot \varepsilon \cdot (\beta/(\beta-1)) k_v q_v\\
    &= 2(\beta/(\beta-1))\varepsilon\cdot\sum_{e\in\delta(v)}{{d_e}{q_v}}\
        =\ 2(\beta/(\beta-1))\varepsilon\cdot\sum_{e\in\delta(v)}{\pi_e}.
  \end{align*}
  \item[Case 2] $\lceil\sum_{e\in\delta(v)}{d_e}/k_v\rceil\ =  1$:   In this case, we pick one copy of vertex $v$, whose cost is:
  \begin{align*}
  c_v&\leq \varepsilon\cdot W_v \\
  &\leq \varepsilon \cdot \sum_{e \sim v} \pi_e\ =\ \varepsilon \cdot \left(\sum_{e\in\delta(v)}{\pi_e}+\sum_{e\notin\delta(v),\, e \sim v}{\pi_e} \right),\\
\end{align*}
\end{description}
In this case, though every edge multiplies its own demand, the selected copies also multiplies the same constant. Through the analysis we have verified that even each edge has its own demand,
the approximation ratio of the proposed algorithm remains unchanged as it is for the uniform case, in the static setting.

\smallskip

\noindent
{\bf Open problems.}
Unfortunately, when we consider the dynamic operation, an edge insertion or deletion may cause a vertex to adjust its level severely because the edge weight in this case connects to its demand.
Using analogous arguments in this paper fails to bound the amortized time for the update event.
It is open whether we can maintain
a constant approximation ratio
with polylogarithmic update time for this general problem where edges have non-uniform unsplittable demands.

\smallskip

\noindent
{\bf Two simple alternatives.}  Here, however, we still present two simple approaches for this problem by combining other techniques with the original proposed level scheme.

The first approach is to partition all of edges into $\log_2 (d_{\rm max})$ clusters according to its demand (where the $i$th cluster contains edges with demand in the range $[2^{i-1}, 2^i)$),
and maintain each cluster by its own data structure.
In every cluster, we set value of $\alpha = 2((2\beta+1)/\beta+2\epsilon)$.
Every time where there is an edge insertion or edge deletion, we run the proposed algorithm in the corresponding cluster.
That is, only the data structure of one cluster is updated per each edge update event.
For the output, we simply select the vertices, and their corresponding number of copies, in each of the cluster to cover all the edges in that cluster.
After these changes, we obtain an $O(\log d_{\rm max})$ approximation ratio solution with $O(L/\epsilon)=O(\log k_{\rm max}/\epsilon)$ update time, where $d_{\rm max}= \max_e\{d_e\}$, $k_{\rm max} = \max_v\{k_v\}$.

The second approach works for integral demands.  We will view an edge $e$ with demand $d$ as $d$ edges $e_1, e_2, \ldots, e_d$ with uniform demand
between the same endpoints.  Then, we will execute the proposed level scheme.   The only problem is that those
edges $e_1, e_2, \ldots, e_d$ corresponding to the original edge $e$ may be assigned to the different endpoints.
What we will do is simply assign \emph{all} edges to the endpoint that is covering the majority of these edges, based on the solution in the proposed level scheme.
After that, the total cost will be increased by at most a factor of $2$, so that
we obtain an O(1)-approximate solution with the $O(d_{\rm max}L/\epsilon)=O(d_{\rm max}\log k_{\rm max}/\epsilon)$ amortized update time.

\section{Concluding Remarks}
We have extended dynamic vertex cover to the more general WMCVC problem, and
developed a constant-factor dynamic approximation algorithm
with $O(\log n / \epsilon)$ amortized update time, where $n$ is the number of the vertices.
Note that, with minor adaptions to the greedy algorithm reported in Gupta et al.'s very recent paper~\cite{GKKP} is
also able to work for the dynamic capacitated vertex cover problem,
but only to obtain a \emph{logarithmic}-factor approximation algorithm with
$O(\log n)$ amortized update time.
Moreover, our proposed algorithm can also be extended to solve the (soft) capacitated set cover problem,
and the capacitated vertex cover problem with non-uniform unsplittable edge demand.

We conclude this paper with some open problems.
First,
recall that in the \emph{static} model,  the soft capacitated vertex cover problem~\cite{GHKO}
can be approximated within a factor of two and three
for the uniform and non-uniform edge demand cases, respectively.
Here, we have shown that it is possible to design
a dynamic scheme with $O(1)$ approximation ratio
with polylogartihmic update time for the uniform edge demand case.
Thus, designing an $O(1)$-approximation ratio algorithm with $O(\log k_{\rm max})$, or polylogarithmic, update time
for the non-uniform edge demand case seems promising.

Moreover, it would also be of significant interest to explore whether it is possible to derive a constant approximation ratio
for the WMCVC problem under constant update time.
Also, in recent years, more studies on the worst-case update time for dynamic algorithms have been conducted.
It would be worthwhile to examine update time in the worst-case analysis.


\newpage

\begin{thebibliography}{99}

\bibitem{AT}
A. Andersson and M. Thorup. Dynamic ordered sets
with exponential search trees.
Journal of the ACM (JACM), Vol. 54, Issue 3, No. 13 , 2007.

\bibitem{BGS}
S. Baswana, M. Gupta, and S. Sen.
Fully dynamic maximal matching in $O(\log n)$ update time.
SIAM J. Comput. 44(2015), no. 1, pp. 88--113.

\bibitem{BCH1} 
S. Bhattacharya, D. Chakrabarty, and M. Henzinger.
Fully dynamic approximate maximum matching and minimum vertex cover in $O(\log^3 n)$ worst case update time. In Proc.
the 28th
ACM-SIAM Symposium on Discrete Algorithms (SODA), Barcelona, Spain, 2017, pp. 470--489.

\bibitem{BCH2}
S. Bhattacharya, D. Chakrabarty, and M. Henzinger.
Deterministic fully dynamic approximate vertex cover and fractional matching in $O(1)$ amortized update time. In Proc.
the 19th Conference on Integer Programming and Combinatorial Optimization (IPCO), Waterloo, Canada, 2017, pp. 86--98.

\bibitem{BHI1}
S. Bhattacharya, M. Henzinger, and G. F. Italiano. Deterministic fully dynamic data structures for vertex cover and matching. In Proc.
the 26th
ACM-SIAM Symposium on Discrete Algorithms (SODA), Philadelphia, USA, 2015, pp. 785--804.

\bibitem{BHI2}
S. Bhattacharya, M. Henzinger, and G. F. Italiano. Design of dynamic algorithms via primal-dual method. In Proc.
the 42nd International Colloquium
on Automata, Languages, and Programming (ICALP), Heidelberg, Germany 2015, pp. 206--218.

\bibitem{DI}
C. Demetrescu and G. F. Italiano. A new approach to
dynamic all pairs shortest paths.
Journal of the ACM (JACM), Vol. 51, Issue 6, 2004, pp. 968--992.

\bibitem{GHKO}
S. Guha, R. Hassin, S. Khuller, and E. Or. Capacitated vertex covering.
Journal of Algorithms, Vol. 48, Issue 1, August 2003, pp. 257--270.

\bibitem{GKKP}
A. Gupta, R. Krishnaswamy, A. Kumar, and D. Panigrahi.
Online and dynamic algorithms for set cover. In Proc.
the 49th ACM Symposium on Theory of Computing (STOC), Montreal, Canada, 2017, pp. 537--550.

\bibitem{HL}
M. T. J. Holm, K. de. Lichtenberg. Poly-logarithmic deterministic fully-dynamic algorithms for
connectivity, minimum spanning tree, 2-edge, and biconnectivity.
Journal of the ACM (JACM) Vol. 48 Issue 4, 2001, pp. 723--760.

\bibitem{ZL}
Z. Ivkovic and E. L. Lloyd. Fully dynamic maintenance of vertex cover. In Proc.
the 19th
International Workshop on Graph-theoretic Concepts in Computer Science (WG),
London, UK, 1994, pp. 99--111.

\bibitem{NS}
O. Neiman and S. Solomon. Simple deterministic algorithms for fully dynamic maximal matching. In Proc.
the 45th ACM Symposium on Theory of Computing (STOC), Palo Alto, USA, 2013, pp. 745--754.

\bibitem{OR}
K. Onak and R. Rubinfeld. Maintaining a large matching and a small vertex cover. In Proc.
the 42nd ACM Symposium on Theory of Computing (STOC), Cambridge, USA, 2010, pp. 457--464.

\bibitem{PS}
D. Peleg and S. Solomon. Dynamic (1+$\epsilon$)-approximate matchings: a density-sensitive approach. In Proc.
the 27th
ACM-SIAM Symposium on Discrete Algorithms (SODA), Virginia, USA, 2015, pp. 712--729.

\bibitem{SS}
S. Solomon. Fully dynamic maximal matching in constant update time. In Proc.
the 57th
Symposium on Foundations of Computer Science (FOCS), New Jersey, USA, 2016, pp. 325--334.



\end{thebibliography}
\end{document}